\documentclass[a4paper,11pt]{article}
\usepackage{amsmath}
\usepackage{amssymb}
\usepackage{url}
\usepackage{algorithm2e}
\usepackage{graphicx}
\usepackage{multicol}
\usepackage{latexsym}
\usepackage{times}

\usepackage{amscd}
\usepackage{subcaption}
\usepackage{float}

\usepackage{tikz}

\usepackage[margin=1cm]{caption}

\usepackage{authblk}

\title{Sex as Gibbs Sampling: a probability model of evolution}
 
  \author{Chris Watkins \thanks{C.J.Watkins@rhul.ac.uk}}
 \author{Yvonne Buttkewitz \thanks{Yvonne.Buttkewitz@rhul.ac.uk}} 
    \affil {Department of Computer Science\\
    Royal Holloway, University of London\\
    Egham, Surrey TW20 0EX, UK}

\newcommand{\BlackBox}{\rule{1.5ex}{1.5ex}}  
\newenvironment{proof}{\par\noindent{\bf Proof\ }}{\hfill\BlackBox\\[2mm]}

\newtheorem{theorem}{Theorem}

\newtheorem{proposition}[theorem]{Proposition}

\newcommand{\breedprob}{p_B}

\newcommand{\fitness}{F}

\newcommand{\genomes}{\mathbf{G}}
\newcommand{\given}{\vert}
\newcommand{\poprv}{\mathrm{G}}

\newcommand{\stationary}{\pi}

\begin{document}

\maketitle

\begin{abstract}
We show that evolution can be modelled as fitting a probability model using standard Markov-chain Monte-Carlo (MCMC) sampling. With some care, `genetic algorithms' can be constructed that are reversible Markov chains that satisfy detailed balance; it follows that the stationary distribution of populations is a Gibbs distribution in a simple factorised form. For some standard and popular nonparametric probability models, we exhibit Gibbs-sampling procedures that are plausible genetic algorithms.  At mutation-selection equilibrium, a population of genomes is analogous to a sample from a Bayesian posterior, and the genomes are analogous to latent variables.  We suggest this is a general, tractable, and insightful formulation of evolutionary computation in terms of standard machine learning concepts and techniques.

In addition, we show that evolutionary processes in which selection acts by differences in fecundity are not reversible, and also that it is not possible to construct reversible evolutionary models in which each child is produced by only two parents. 
\end{abstract}

\section{Introduction: evolution as learning}

Evolution can be viewed as a learning algorithm, in that populations evolve improved solutions to the problem of how to survive and reproduce.  These improvements are the result of empirical experience in the sense that each organism is an experiment that is evaluated according to its lifetime reproductive success. The learning happens over many lifetimes. Can evolution be connected in a natural way to existing learning algorithms, or is it a learning process of a different type? We propose some  probability models of evolution that are analogous to Bayesian inference.  For these models, straightforward MCMC methods can also be viewed as evolutionary computation. These models are not toys: they are fully implementable as evolutionary computations for arbitrary fitness functions. We exhibit evolutionarily plausible sampling procedures, and we take care to show that the resulting Markov chains satisfy detailed balance, so that the mutation-selection equilibrium distribution over populations is precisely a Gibbs distribution. These sampling procedures are fully implementable for arbitrary fitness functions.  

In this paper we consider only evolution in sexual populations. This is partly because our approach can be applied most easily to sexual evolution, but also because it is sexual evolution that is most interesting for machine learning because it has produced the most complex organisms with the fewest lifetimes. Nearly all multicellular organisms have evolved sexually, as discussed by \cite{smith1989evolutionary}: all of the evolution of complex innate behaviour and of elaborate neural systems has been in sexual populations.

It is natural to model evolution as a Markov process. Suppose that evolution proceeds under constant conditions, so that there is a sequence of finite populations $\poprv^{1}, \poprv^{2}, \ldots , \poprv^{t}, \ldots$  Successive populations may have individuals in common: indeed, we will mostly consider algorithms in which at most one member of the population is replaced at each time-step. In genetics, these are known as `overlapping generations' models.  Each population $\poprv^{t+1}$ is produced by breeding and selection from the previous population $\poprv^{t}$. Since by assumption the environmental conditions that determine selection are constant, and breeding depends on the current population only, it follows that the sequence of populations is a Markov chain. 

In evolutionary computation it is typically convenient to require that each successive population has the same number $N$ of individuals; this ensures that the artificial population does not grow to unbounded size, and never goes extinct.  For any reasonable definition of mutation, there is a small but finite chance of any genome mutating into any other, so that any population can mutate into any other.  It follows that the Markov chain of populations is aperiodic and irreducible. One can construct models in which genomes vary in size, but in any reasonable such model there is some selective advantage in reducing genome size, so that the Markov chain will be positive recurrent. It follows that the sequence of populations will converge to a unique stationary distribution, which is also known as the mutation-selection equilibrium. 

One basic goal of evolutionary theory is to characterise this stationary distribution. For the standard models of evolution used in evolutionary computation and in population genetics, it is only possible to characterise the stationary distribution for special cases of the fitness function. The reason is that these algorithms generate  Markov chains that are slightly too complicated for easy analysis, and in particular they do not satisfy detailed balance, which is a property that can make the stationary distribution much easier to characterise. 

We present a class of simple models of evolution with Markov chains that satisfy detailed balance, and for which the stationary distribution factorises as: 
\begin{equation}
\label{eq:stationary}
\pi(\poprv) \propto p_B( \poprv ) \fitness( \poprv ) 
\end{equation}
where $\pi(\cdot)$ is the stationary distribution of populations;  $p_B(\cdot)$ is the probability of breeding the population with no selection; and $\fitness(G)$ is the fitness of the population.  $p_B(G)$  is analogous to the prior probability of the population.
A simple definition of $\fitness(\poprv)$ is as the product of the fitnesses of the individuals in the population; this is analogous to likelihood.   With this definition of population fitness, and assuming populations of $N$ individuals, equation\ \ref{eq:stationary} becomes
\begin{equation*}
p(\poprv) \propto p_B( \poprv ) \prod_{i=1}^{N}f(\poprv_i)
\end{equation*}
The form of this expression is similar to that for a Bayesian posterior distribution; indeed, in our models a population of genomes is closely analogous to a single sample of parameter values from a Bayesian posterior. 


\section{Exchangeable breeding} 
\label{sec:ExchangeableBreeding}

We model breeding using an exchangeable probability distribution $p_B(\cdot)$ over genomes. In our algorithms, a genome is never sampled \emph{ab initio}, but is always sampled conditionally given an existing population of genomes. Given a population $G$  comprising  $N$ genomes $(g_1, \ldots, g_N)$, we may: 
\begin{itemize}
\item Gibbs-sample a new $i$th genome $ g_i \sim p_B( \cdot \given G_{-i})$, where $G_{-i} = G \setminus g_i $, or
\item breed additional children $c_1, \ldots, c_M$ sequentially by $c_{i+1} \sim p_B(\cdot \given G, c_1, \ldots, c_i)$
\end{itemize}
The effect of exchangeability is to allow the interchange of ancestors and descendants, which allows construction of a reversible Markov chain of populations. The exchangeable breeding (EB) models we suggest are some of the most popular probability models in machine learning.

\subsection{The Dirichlet process as an exchangeable breeding model}
\label{sec:DirichletBreeding}

Suppose each individual consists of a single `locus' -- one atomic genetic element, which can mutate into different `alleles'. We temporarily write $g_i = \theta_i $, to emphasise that each genome consists of a single `allele', and so that we use the usual notation for Dirichlet processes.   The alleles are analogous to latent parameters.

Consider first breeding without any selection. That is, we suppose at first that all individuals are fit, so that all survive, and we need not consider the complications resulting from selection.   One generation of breeding is as follows. A randomly chosen genome $\theta_i$ `dies', and is replaced by another genome bred from the rest of the population.  The new allele bred -- denoted by $\tilde{\theta}$ is either an exact copy of one of the other alleles in the population, or else it is a `mutation'. For the moment we model mutation as generating a new, independent value, sampled from a `base distribution' $H$.  A mutation occurs with probability $u > 0 $. The breeding algorithm is: 
\begin{center}
\begin{minipage}{12cm}
\begin{algorithm}[H]
\SetAlgoCaptionLayout{centerline}
\DontPrintSemicolon 
\Repeat{forever}
{
{Randomly select an index $i$, $1 \le i \le N$}

{$\tilde{\theta} \sim (1-u) \textrm{uniform}( \theta_{-i} )  + u H$}\\
{$\theta_i \leftarrow \tilde{\theta}$}
 } 

\caption{Genetic algorithm with only one gene and no selection}
\label{alg:DirMutation}
\end{algorithm}
\end{minipage}
\end{center}

\noindent This is a variant of the Moran process \cite{moran1962statistical}, which is a well known model in population genetics. Depending on what aspects of genetic processes are being modelled, the alleles $\theta$ may limited to a small set of discrete values, an unbounded set of discrete values, or may be drawn from a continuous probability distribution. The Moran process is well known to be a reversible Markov chain \cite{ewens2004mathematical}.
In this simplest formulation, a mutant value is independent of the existing allele values. This a well-known simplifying assumption in population genetics, that was termed the `house of cards' model by \cite{kingman1978simple}.  

The impatient reader will already have noticed that algorithm\ \eqref{alg:DirMutation} is identical to the procedure for Gibbs sampling from the predictive distribution of a Dirichlet process according the the Blackwell-MacQueen urn model \cite{blackwell1973ferguson},:

\begin{center}
\begin{minipage}{12cm}
\begin{algorithm}[H]
\SetAlgoCaptionLayout{centerline}
\DontPrintSemicolon 
\Repeat{forever}
{
{Randomly select an index $i$, $1 \le i \le N$}\\
{$\tilde{\theta} \sim \frac{N-1}{N-1+\alpha}~ \textrm{uniform}( \theta_{-i} )  + \frac{\alpha}{N-1+\alpha}~ H$}\\
{$\theta_i \leftarrow \tilde{\theta}$}
 } 

\caption{Genetic algorithm with only one gene and no selection}
\label{alg:BlackwellMacQueen}
\end{algorithm}
\end{minipage}
\end{center}

\noindent The only difference between the `breeding' in algorithm\ \ref{alg:DirMutation} and the Gibbs sampling in algorithm\ \ref{alg:BlackwellMacQueen} is that the mutation rate $u$ is re-parameterised in terms of the concentration parameter $\alpha$, so that:
\begin{align}
\label{eq:uandalpha}
u & = \frac{\alpha}{N-1+\alpha}, \quad    \textrm{and equivalently} &    \alpha = (N-1) \frac{u}{1-u}
\end{align}
\noindent This similarity between the Moran process and the Blackwell-MacQueen urn process is hardly an accident -- applications in genetics have motivated many developments of the theory of exchangeable distributions on partitions and of Dirichlet processes: the famous Ewens sampling distribution \cite{ewens1972} is the joint distribution of the allele-counts in a Dirichlet process, and this was derived before Dirichlet processes were formulated by \cite{ferguson1973bayesian}.  The main aim of mathematical and statistical genetics was to develop methods of \emph{analysing} biological genetic data: our aim, in contrast, is to use the probability models as the basis for evolutionary algorithms.

\subsection{Exchangeable mutations} 

A limitation of Blackwell-MacQueen sampling  as a model of breeding  a new allele is that new mutations are independent of existing alleles. It is more reasonable to suppose that new mutations will typically be similar to some existing alleles.

The base distribution $H$ of the Dirichlet process can be an exchangeable distribution, so that new mutations depend on which mutations have been previously accepted.  A natural  distribution for mutations is the Dirichlet diffusion tree \cite{neal2003density}. While still not fully accurate, this more closely approximates the underlying biology.

\subsection{Product of DP models}
\label{sec:productDPmodel}

Now extend the previous model so that each genome is a vector of $L$ elements. A population of $N$ genomes is therefore represented as an $N$ by $L$ array of allele values, that we denote by $\genomes = (\theta_{ij})_{1\le i \le N, ~ 1\le j \le L}$. Each row of this array is a `genome'. Each column consists of the allele values at a particular locus on each of the genomes.  Breeding -- which is the same as Gibbs sampling -- is carried out by picking a row to resample, and then Gibbs sampling each element in that row. Each element of the newly bred genome is Gibbs-sampled from its own column distribution, independently of the other elements of the row. This breeding model is therefore equivalent to Gibbs-sampling corresponding elements from a product of $L$ Blackwell-MacQueen urn processes. 

\noindent This product-of-DPs breeding model is one example of an exchangeable breeding (EB) model that is similar to genetic algorithms. More elaborate EB models can be constructed to model genomes that are variable size sets of interacting elements, but these are beyond the scope of this paper.

\section{Modelling Selection} 

We set out four approaches to modelling natural selection: rejection sampling; Metropolis-Hastings; fitness as expected lifespan in continuous time; and tournaments between many children. All four methods satisfy detailed balance and lead to the same closed-form factorisation of the stationary distribution as proportional to breeding probability times fitness. 

None of these models of selection is strictly biologically accurate -- but all four methods allow differences in calculated fitness of an individual to affect its inclusion in the population, in a manner that is at least as plausible as the selection procedures used in evolutionary computation and in the simpler models of mathematical genetics. 

\subsection{Selection by rejection sampling} 
\label{sec:selectionrejection}

For any genome $g$, let the fitness $f(g) \in (0,1]$. Let the population be $G = (g_1, \ldots, g_N)$, and let $G_{-i}$ denote $(g_1, \ldots, g_{i-1}, g_{i+1}, \ldots, g_N)$. The rejection sampling algorithm is: 

\begin{center}
\begin{minipage}{12cm}
\begin{algorithm}[H]
\SetAlgoCaptionLayout{centerline}
\DontPrintSemicolon 
\Repeat{forever}
{
{randomly select an index $i$, $1 \le i \le N$}

\Repeat{acceptance}{
{propose $g_i^\prime \sim p_B( \cdot \given G_{-i})$}
 
{accept $g_i \leftarrow g_i^\prime $ with probability $f(g_i^\prime)$}
 } 
}

\caption{Selection by rejection sampling}
\label{alg:rejection}
\end{algorithm}
\end{minipage}
\end{center}

The following proposition is well known and part of MCMC folklore, but we sketch a proof here for completeness. 
\begin{proposition}
\label{prop:rejection}
The stationary distribution $\stationary$ of algorithm\ \ref{alg:rejection} is 
$\stationary(\poprv) =  \frac{1}{Z_\pi} \breedprob(\poprv) \prod_{i=1}^{N} f(G_i)$, where $Z_\pi$ is a normalising constant. 
\end{proposition}
\begin{proof}
Consider a distribution $\pi$ over populations, such that $\stationary(\poprv) \propto \breedprob(\poprv) \prod_{i=1}^{N} f(G_i)$. We shall show that this distribution together with the state-transitions produced by algorithm \ref{alg:rejection} satisfies  detailed balance. From this, it follows that $\pi$ is the stationary distribution. 

Consider state $G = (g_1, g_2, \ldots, g_i, \ldots, g_N)$;  let genome $i$ be resampled, and consider the probability $p( G \rightarrow G^\prime )$of a state-change to $G^\prime = (g_1, g_2, \ldots, g_i^\prime, \ldots, g_N)$, given that the current state is $G$.   Then: 
\begin{align*} 
p( G \rightarrow G^\prime ) &=  \frac{1}{Z_T(G_{-i})} ~ p_B(G_i^\prime \given G_{-i}) f(G^\prime)\\
\intertext{so}
\pi(G) p( G \rightarrow G^\prime ) &= \frac{1}{Z_\pi} ~ p_B(G_i \given G_{-i}) p_B(G_{-i}) f(G)p( G \rightarrow G^\prime ) \\
& = \frac{1}{Z_\pi Z_T(G_{-i})} ~ p_B(G_i \given G_{-i}) p_B(G_{-i}) f(G)p_B(G_i^\prime \given G_{-i}) f(G^\prime)\\
&= \frac{1}{Z_\pi Z_T(G_{-i})} ~ p_B(G_i \given G_{-i})p_B(G_i^\prime \given G_{-i})~p_B(G_{-i}) ~ f(G) f(G^\prime)\\
\intertext{which is symmetric between $G$ and $G^\prime$, so}
&= \pi(G^\prime) p(G^\prime \rightarrow G)
\end{align*} 
Hence detailed balance holds for $\pi$, and it follows that $\pi$ is the stationary distribution.
\end{proof}

\noindent Rejection sampling is inefficient if the expected value of $f$ is small; a Gibbs-within-Metropolis algorithm  is more efficient, and  does not require $f$ to be bounded above:

\begin{center}
\begin{minipage}{12cm}
\begin{algorithm}[H]
\SetAlgoCaptionLayout{centerline}
\DontPrintSemicolon 

\Repeat{forever}{ 

{randomly select $i$ from $\{1, 2, \ldots, N\}$}

{propose $ g_i^\prime \sim p_B(\cdot \given G_{-i})$}

{accept $g_i \leftarrow g_i^\prime$ with probability $\min \biggl\{1, \dfrac{f(g^\prime)}{f(g)}\biggr\}$}
}

\caption{Gibbs-within-Metropolis}

\label{alg:GibbsMH}
\end{algorithm}
\end{minipage}
\end{center}

\noindent Algorithm\ \ref{alg:GibbsMH} has the same stationary distribution as algorithm\ \ref{alg:rejection}, but it does not describe a plausible breeding and selection model.  A child $g_i^\prime$ is bred from all individuals except for $g_i$,  and $g_i^\prime$ is then accepted if it wins a tournament with $g_i$, the only individual that is not its parent: this is not a reasonable biological story.  A more natural selection model is subset selection below.

\subsection{Subset selection} 
\label{sec:selectionSR}

\newcommand{\bfs}{\textbf{s}}
\newcommand{\bfr}{\textbf{r}}

Consider a breeding and selection cycle that starts with a population $G = (g_1, \ldots, g_N)$ . from which a set of $M$ new children  $C=(c_1, \ldots, c_M)$ are bred. This breeding is done in an exchangeable manner, in that children are bred one at a time, their parents being $G$ together with the preceding children, so that
 \begin{equation}
 \label{eq:sequentialBreeding}
 c_{i+1} \sim p_B(\cdot \given G, c_1, \ldots, c_i)
 \end{equation}
 From the combined population of size $N+M$,  a subset of size $N$ is  selected to form the new population $G^\prime$.  $G^\prime$ is sampled from a fitness-weighted distribution over subsets  of $G \cup C$.   Let $F$ be a fitness function that applies to \emph{sets} of genomes. One possible definition of $F(G)$ for a set $G$ of genomes is: 
\begin{equation*}
\label{eq:fS}
F( G ) = \prod_{g \in G} f(g)
\end{equation*}
but in general $F$ may be any non-negative function of genome-sets of any size. $F$ may depend only upon the fitnesses of individual genomes in the set, or it may express the joint fitness of a coalition of cooperators.

Let $U = G\cup C$. In the simplest case, suppose that one alternative subset $G^\prime \subset U$ is proposed and considered for acceptance. Let the probability of proposing $G^\prime$ be $q(G^\prime \given G, U)$, where $q$ is a proposal distribution such that for all $G$, $G^\prime$, 
$q(G\given G^\prime, U) = q( G^\prime \given G, U)$. A possible choice of $q$ is for $G^\prime$ to be a random subset of $U$.

\noindent The sampling algorithm is then: 
\begin{center}
\begin{minipage}{12cm}
\begin{algorithm}[H]
\SetAlgoCaptionLayout{centerline}
\DontPrintSemicolon 
\KwData{{$G$: starting population}\\
{$M$: number of children to breed}\\
{$K$: number of subset proposals}}
\Repeat{forever}{ 
\nlset{Breeding:}\For{$i \leftarrow 1$ \KwTo $M$}{Sample $c_{i}  \sim p_B(\cdot \given g_1, \ldots g_{N}, c_1, \ldots c_{i-1})$}
{Let $U= G \cup \{c_1, \ldots, c_M \}$}\\
\nlset{Selection:}\For{$K$ times}{
{Propose $G^\prime \sim q(\cdot \given G, U)$}\\
{Set $G \leftarrow G^\prime$ with probability $\min\biggl\{ 1, \dfrac{F(G^\prime)}{F(G)}\biggr\}$}
}
}
\caption{Subset selection}
\label{alg:SubsetSelection}
\end{algorithm}
\end{minipage}
\end{center}

\noindent There are a number of possible biological interpretations of this algorithm.  If $M=1$, this is a Moran process with selection, and $K$ members of the population are proposed for elimination; the least fit individual is most likely to be replaced. If $M$ is as large as $N$, the algorithm is somewhat similar to a Wright-Fisher process, but with two crucial differences: first, the children are bred sequentially according to equation\ \eqref{eq:sequentialBreeding}; second, there is stochastic partial overlap between successive populations.  $K$ coalitions of individuals are sequentially compared for acceptance; the fittest coalition is most likely to be accepted.

\smallskip

\noindent Algorithm\ \eqref{alg:SubsetSelection} satisfies detailed balance:

\begin{proof}
We take $K$ to be 1; multiple repetitions of the selection loop will leave the stationary distribution invariant. Let $\pi$ be defined as
\begin{equation*}
\pi(G) = \frac{1}{Z_\pi}  p_B(G) F(G) 
\end{equation*}
Consider the transition probability from $G$ to $G^\prime$ in two stages: the probability of breeding genomes $C$ to construct the population $U = G \cup C$; and the probability of transitioning from $G$ to $G^\prime$ within $U$. Let $C^\prime = U \setminus G^\prime$, so that $U = G \cup C = G^\prime \cup C^\prime$.  Then, assuming w.l.o.g. that $F(G^\prime) < F(G)$, we have:
\begin{align*}
p( G \rightarrow G, C ) & = p_B(C \given G) \\
p( G, C \rightarrow G^\prime, C^\prime ) & =  q(G^\prime \given G, U) \frac{F(G^\prime)}{F(G)}\\
p( G^\prime, C^\prime \rightarrow G, C ) & =  q(G \given G^\prime, U)
\end{align*}
By assumption, $q(G^\prime \given G, U) = q( G \given G^\prime, U)$.
It follows that 
\begin{align*}
\pi(G) p( G \rightarrow G^\prime, C^\prime ) & = p_B( G ) F(G)  p_B(C \given G ) q(G^\prime \given G, U) \frac{F(G^\prime)}{F(G)}\\
&= p_B(U) F(G^\prime) q(G^\prime \given G, U) & &\mbox{by exchangeability of $p_B$;}\\
&= p_B(G^\prime \cup C^\prime) F(G^\prime) q(G^\prime \given G, U)& &\mbox{since $G\cup C = G^\prime \cup C^\prime$;}\\
&= p_B(G^\prime) F(G^\prime) p_B(C^\prime \given G^\prime) q(G \given G^\prime, U)& &\mbox{by symmetry of $q$;}\\
&= \pi(G^\prime) p( G^\prime \rightarrow G, C )
\end{align*}
Hence $\pi$ is the stationary distribution.
\end{proof}

\subsection{Fitness as expected lifetime}
\label{sec:selectionlifetime}
Another approach to selection is by preferentially removing individuals of low fitness, so that the unfit `die' sooner, but all individuals contribute equally to reproduction whilst they are `alive'. Consider a Markov process in continuous time, and once again let $f$ be a strictly positive real-valued function of genomes.  For each individual $g_i$, let $\frac{1}{f(g_i)}$ be the hazard rate of mortality; that is, the probability that $g_i$ will `die' during the interval of length $dt$ is $\frac{dt}{f(g_i)}$.  When an individual dies, it is immediately replaced by another individual bred from the rest of the population.  We suppose that only one individual `dies' at a time, and its replacement is instantly bred from the remaining $N-1$ individuals. To avoid pathological situations in which there is an unbounded number of breeding events, we require that the expected lifetime of accepted individuals is greater than zero. 
 
\begin{proof}
Suppose as before that the population is $G= (g_1, \ldots, g_N)$. Define
\begin{equation*}
\pi(G) = \frac{1}{Z_\pi} p_B(G) \prod_{j=1}^N f(g_j)
\end{equation*}
and let  $h(G \rightarrow G^\prime)$ be the hazard rate of a transition to $G^\prime = (g_1, \ldots, g_i^\prime, \ldots, g_N)$   
\begin{equation*}
h( G \rightarrow G^\prime )  = \frac{1}{f(g_i)} p_B( g_i^\prime \given G_{-i} )
\end{equation*}
It follows that
\begin{align*}
\pi( G ) h( G \rightarrow G^\prime )& = \frac{1}{Z_\pi}p_B(G) \left( \prod_{j=1}^N f(g_j) \right)~\frac{1}{f(g_i)} p_B( g_i^\prime \given G_{-i} )\\
&= \frac{1}{Z_\pi} \quad  p_B(g_i\given G_{-i}) p_B(g_i^\prime\given G_{-i})~p_B(G_{-i})  \quad \prod_{j \neq i} f(g_j)
\end{align*}
which is symmetric between $g_i$ and $g_i^\prime$, showing that detailed balance holds. 
\end{proof}

\subsection{Tournaments between children only}
\label{sec:TournamentsBetweenChildren}
 It would be equally biologically plausible to arrange a tournament between children only. When an existing member $g_i$  of the population dies,  $K$ children are generated, and the winner of a tournament between the children is selected to replace $g_i$. Let the children be $c_1, \ldots, c_K \sim p_B(\cdot\given G_{-i})$, with fitnesses $\tilde{f}_1,\ldots,\tilde{f}_K$. Note that in this case the children are not bred exchangeably as in subset selection -- they are all breed only from $G_{-i}$. Let the probability that child $c_k$ wins the tournament be $p_W(c_k)$: 
\begin{equation*}
p_W(c_k) = \frac{f(c_k)}{\sum_{j=1}^K f(c_j)}
\end{equation*}
For large enough  $K$, this method of selecting the winning child will become close to Gibbs-sampling $g_i$ from a distribution proportional to $p_B(\cdot\given G_{-i} ) f(\cdot)$ ; that is, a tournament between sufficiently many children becomes arbitrarily close to Gibbs-sampling $g_i$ from $\pi(\cdot\given G_{-i})$. However, for small values of $K$, a selection tournament between $K$ children is only an approximation to Gibbs sampling, as noted by \cite{neal2000markov} with reference to approximate Gibbs sampling in Dirichlet mixture models. 

\subsection{Multiple fitness functions}

In all of the above algorithms, we may allow the fitness function to depend upon the index of the genome in the population; that is, we may specify the fitness of a population as:
\begin{equation*}
f(G) = \prod_{i=1}^L f_i(g_i)
\end{equation*}
where $f_i$ is a fitness function for the $i$th `niche' in the population.  One fundamental difference between sexual and asexual populations is that in a sexual population, different organisms -- and indeed different sub-populations -- may experience different selective pressures, and yet their genomes are recombined, so that the multiple selective pressures affect the entire population.  This provides an additional level of information flow from the environment to the genomes. A plate diagram of the resulting probability model is in figure \ref{fig:plateDiagramWithFitness} below. 

%
%
%
%

\usetikzlibrary{shapes}
\usetikzlibrary{fit}
\usetikzlibrary{chains}
\usetikzlibrary{arrows}

\tikzstyle{latent} = [circle,fill=white,draw=black,inner sep=1pt,
minimum size=20pt, font=\fontsize{10}{10}\selectfont, node distance=1]
\tikzstyle{obs} = [latent,fill=gray!25]
\tikzstyle{const} = [rectangle, inner sep=0pt, node distance=1]
\tikzstyle{factor} = [rectangle, fill=black,minimum size=5pt, inner
sep=0pt, node distance=0.4]
\tikzstyle{det} = [latent, diamond]

\tikzstyle{plate} = [draw, rectangle, rounded corners, fit=#1]
\tikzstyle{wrap} = [inner sep=0pt, fit=#1]
\tikzstyle{gate} = [draw, rectangle, dashed, fit=#1]

\tikzstyle{caption} = [font=\footnotesize, node distance=0] %
\tikzstyle{plate caption} = [caption, node distance=0, inner sep=0pt,
below left=5pt and 0pt of #1.south east] %
\tikzstyle{factor caption} = [caption] %
\tikzstyle{every label} += [caption] %

\tikzset{>={triangle 45}}


\newcommand{\factoredge}[4][]{ %
  \foreach \f in {#3} { %
    \foreach \x in {#2} { %
      \draw[-,#1] (\x) edge[-] (\f) ; %
    } ;
    \foreach \y in {#4} { %
      \draw[->,#1] (\f) -- (\y) ; %
    } ;
  } ;
}

\newcommand{\edge}[3][]{ %
  \foreach \x in {#2} { %
    \foreach \y in {#3} { %
      \draw[->,#1] (\x) -- (\y) ;%
    } ;
  } ;
}

\newcommand{\factor}[5][]{ %
  \node[factor, label={[name=#2-caption]#3}, name=#2, #1,
  alias=#2-alias] {} ; %
  \factoredge {#4} {#2-alias} {#5} ; %
}

\newcommand{\plate}[4][]{ %
  \node[wrap=#3] (#2-wrap) {}; %
  \node[plate caption=#2-wrap] (#2-caption) {#4}; %
  \node[plate=(#2-wrap)(#2-caption), #1] (#2) {}; %
}

\newcommand{\gate}[4][]{ %
  \node[gate=#3, name=#2, #1, alias=#2-alias] {}; %
  \foreach \x in {#4} { %
    \draw [-*,thick] (\x) -- (#2-alias); %
  } ;%
}

\newcommand{\vgate}[6]{ %
  \node[wrap=#2] (#1-left) {}; %
  \node[wrap=#4] (#1-right) {}; %
  \node[gate=(#1-left)(#1-right)] (#1) {}; %
  \node[caption, below left=of #1.north ] (#1-left-caption)
  {#3}; %
  \node[caption, below right=of #1.north ] (#1-right-caption)
  {#5}; %
  \draw [-, dashed] (#1.north) -- (#1.south); %
  \foreach \x in {#6} { %
    \draw [-*,thick] (\x) -- (#1); %
  } ;%
}

\newcommand{\hgate}[6]{ %
  \node[wrap=#2] (#1-top) {}; %
  \node[wrap=#4] (#1-bottom) {}; %
  \node[gate=(#1-top)(#1-bottom)] (#1) {}; %
  \node[caption, above right=of #1.west ] (#1-top-caption)
  {#3}; %
  \node[caption, below right=of #1.west ] (#1-bottom-caption)
  {#5}; %
  \draw [-, dashed] (#1.west) -- (#1.east); %
  \foreach \x in {#6} { %
    \draw [-*,thick] (\x) -- (#1); %
  } ;%
}

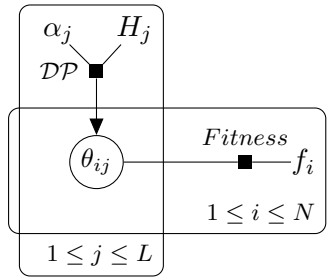
\begin{figure}[ht]
   \begin{center}
\begin{tikzpicture}

\node[latent](theta){$\theta_{ij}$};

\node[const, above=1.2 of theta, xshift=-0.5cm] (alpha) {$\alpha_j$};
\node[const, above=1.2 of theta, xshift=0.5cm] (H) {$H_j$};

\node[const, right=2.2 of theta](f){$f_i$};

\factor[above=0.75 of theta]{theta-f}{left:$\mathcal{DP}$}{alpha,H}{theta};
\factor[right=1.5 of theta]{f-f}{above:${Fitness}$}{f,theta}{};

\node[const, below=0.5 of theta](dummy){$~$};

\plate{loci}{
(theta)(theta-f)(theta-f-caption)
(alpha)
(H)
(dummy)
}{$1 \le j \le L$};

\plate{pop}{
(theta)(f-f)(f-f-caption)
(f)
(loci.west)(loci.east)
}{$1 \le i \le N$};
\end{tikzpicture}
\caption{Plate diagram of the product-of-Dirichlet processes evolutionary model with selection.  Note that each indexed `niche' in the population may have its own fitness function.}
\label{fig:plateDiagramWithFitness}
  \end{center}
\end{figure}

\section{Difficulties}

\subsection{Is it important to model breeding from only two parents?} 

The most obvious difference between the EB models and biological reproduction is that real organisms have only two rather than $N$ parents.  We believe that it is not possible to construct an exact exchangeable breeding model, nor any plausible reversible evolutionary Markov chain, using reproduction from two parents only. 

For a reversible Markov chain in its stationary distribution, it is not possible to tell the direction of time by observing the sequence of states, because any loop of states is equally likely to be traversed in both directions. This is simply the Kolmogorov cycle condition, and it is a basic property of reversible Markov chains.
However, for two-parent reproduction, in interesting parameter regimes, it is possible to identify the direction of time, even at mutation-selection equilibrium. This is because a child will typically be more similar to each of its parents than the parents are similar to each other. 

Consider a population $G$, and suppose two individuals, $a$ and $b$,  are chosen to breed, and they produce a child $c$ that (typically) replaces a distantly related genome $d$, which dies.  After replacement of $d$ by $c$, the population is $G^\prime$. Is it possible to tell whether the transition was from $G$ to $G^\prime$, or from $G^\prime$ to $G$? 

There are two cases. 
\begin{enumerate}
\item If  the parents of $d$ -- let us call them $x$ and $y$ -- are still in the population, then $x$, $y$, $a$ and $b$ are all in the population, and the difference between $G$ and $G^\prime$ is that in $G$ there is a child of $x$ and $y$, and in $G^\prime$ there is instead a child of $a$ and $b$. We cannot tell whether $d$ replaced $c$, or whether $c$ replaced $d$.

\item If one or both of the parents of $d$ have been  eliminated from the population before $c$ is bred from $a$ and $b$, then in $G$ there will be a genome $d$ without two parents, but in $G^\prime$ there will instead be a genome $c$ which has two parents.  We can tell the direction of time. 
\end{enumerate}
It follows that the Markov chain of populations is not reversible if we are able to determine whether or not a genome $c$ does have two parents in the population. We sketch a proof that for sufficiently long binary genomes, under any system of recombination, it is possible to determine whether a pair of genomes $a$ and $b$ are in fact the parents of $c$. 

If $a$ and $b$ are the parents of $c$, then every element of $c$ is a copy of either an element of $a$ or an element of $b$, except for a fraction $u$ of elements that are new mutations. Now consider the set of indices $J$ such that $a_j = b_j$. If the genomes are sufficiently long, $|J|$ will with high probability be large. Let $\hat{u}$ be the fraction of elements of $J$ at which $c_j \neq a_j$. Provided that the genomes are long enough, with high probability $\hat{u}$ will be close to $u$. Furthermore, this will only be the case if $a$ and $b$ are indeed parents of $c$; for any other pair of relatives of $c$, $\hat{u}$ will with high probability be larger than $u$, since there will have been additional copying operations and mutations on any path on the genealogical trees relating each element of $c$ to the corresponding elements in $a$ and $b$. It follows that, for any given population size, provided that $u$ is sufficiently small, and the genomes are sufficiently long, each pair of genomes in the population can be reliably tested to see if they are the parents of $c$. It follows that the direction of time is detectable because we can determine when the parents of genomes in the population disappear. 

Is breeding from only two parents important for evolution? Many researchers have used $N$-parental reproduction in various forms. and where comparisons of $N$- and bi-parental reproduction have been made, such as in \cite{baluja1995removing}, the tentative conclusion has been that $N$-parental reproduction appeared superior for the problems studied.  

On the other hand, it may be an imperfection of biology that each individual has only two parents -- perhaps evolution would be slightly more efficient if we had many parents?  It is hard to imagine a biologically feasible form of sexual reproduction in which the genetic material from multiple parents would be combined in a reliably unbiased way. Is it possible that biological sexual reproduction with only two parents is a sub-optimal implementation of Gibbs sampling?

\subsection{Selection by relative fecundity}
\label{sec:fecundityselection}

Another possible mode of selection is by relative fecundity: individuals have equal life expectancy, but while they are alive they contribute to breeding at a rate proportional to their fitness.     Suppose that $N$ individuals in the population have (positive) fecundity values $w_1, \ldots, w_N$, and suppose that the $i$th individual dies and is to be replaced with $g_i^\prime$, which is `bred' from $G_{-i}$. Let
\begin{equation*}
g_i^\prime \sim p_W \cdot \given G_{-i}, w_{-i} )
\end{equation*} 
where $p_W( \cdot \given G_{-i}, w_{-i})$ is a probability distribution over genomes that depends on the genomes $G_{-i}$ and their fecundities only. 

One possible breeding algorithm is  similar to that considered by \cite{dayan1997using}: each genome contributes reproductively in proportion to its fitness; that is the $j$th element of the Gibbs proposal to replace the $i$th genome is sampled from 
\begin{equation}
g_{ij}^\prime \sim \frac{1}{\alpha + \sum_{k \neq i} w(g_{k:})} \sum_{k\neq i} w(g_{k:})\delta_{g_{kj}}  \quad + \quad \alpha H
\end{equation} 
Appendix A gives a proof by counterexample that the Markov chain for this breeding method need not satisfy the Kolmogorov cycle condition, and so is not in general reversible. 

\subsection{Mutation rate and population size} 

For a fixed population size $N$, in an EBM such as the DP-product model of section \ref{sec:productDPmodel}, the concentration parameter is functionally related to the mutation rate.  In models where population size is variable, the mutation rate will vary if the concentration parameter is held constant as the size of the population changes. 
Natural populations do typically vary greatly in size over time with approximately constant mutation rate, so that the effective value of $\alpha$ increases with population size.  This is a difference between typical natural evolution and these models.

\section{Relationship to other work}

Formal models of evolution have been studied for over a century by researchers in many fields, and there are vast literatures on models of evolution in the fields of mathematical genetics, evolutionary computation, physics,  probability theory and Markov-chain Monte-Carlo (MCMC) methods, and recently a growing literature in computational learning theory. 

Our models are more concrete than those considered for the deep limitative results initiated by \cite{valiant2009evolvability} and \cite{feldman2009robustness}.

The notion of evolutionary computation was introduced in the 1960s by \cite{owens1966artificial} and others; the most influential model was the genetic algorithm introduced by \cite{holland1975adaptation}, which simulated sexual evolution with crossover and recombination between genomes represented as linear arrays.  Variants of genetic algorithms have been widely applied in optimisation problems, and there is a vast literature on this; two notable books are \cite{goldberg1989genetic} and  \cite{reeves2003genetic}. Although Holland's genetic algorithm is a greatly simplified abstraction of biological sexual reproduction, it is still too complicated for convenient analysis, so that its stationary distribution cannot in general be given in a useful closed form, although some progress was made by \cite{vose1999simple}.  

\cite{liu2008monte} and \cite{strens2003evolutionary} proposed reversible recombination operators for GAs, but these were biologically implausible, since two parents were reversibly recombined into two children, and both children had to be accepted.

More recently in evolutionary computation attention has focused on population based algorithms, in which breeding is either from the whole population or from parametric estimates of the population distribution. An early study was \cite{baluja1995removing}; `estimation of distribution algorithms' (EDAs) have been developed by \cite{zhang1999bayesian},\ \cite{pelikan2002survey},  and  others. These are effective optimisation methods using Bayesian heuristics, but the Markov chains of populations they induce are not typically reversible, though they can be made reversible with subset selection. \cite{shapiro2005drift} described a different approach to constructing reversible EDAs by discretising the parameter space and defining explicit transition probabilities. 

The probability models of section\ \ref{sec:ExchangeableBreeding} are taken from  mathematical population genetics and statistics. A recent review of Dirichlet processes is \cite{TehDP2010}.  One of the original motivations for Dirichlet processes was for the analysis of empirical genetic data, but our motivation is the reverse -- we use the models to define sampling algorithms that can be interpreted as evolutionary computation. Gibbs sampling was introduced by \cite{geman:1984}. 

In physics there is a considerable literature on applying the methods of statistical mechanics to evolution. \cite{sella2005application} consider the restricted case of a regime in which only one mutation is under selection at a time. \cite{ao2008emerging}, \cite{mustonen2010fitness}, and \cite{barton2009statistical} all note that diffusion approximations of the type developed by \cite{crow1970introduction} can be reversible and that equilibrium distributions may be computed by statistical mechanical methods. \cite{de2011contribution} is a recent review of evolutionary models developed in physics, and in other research communities. Our work differs in that we propose implementable sampling algorithms for general fitness functions, for finite populations, which satisfy detailed balance exactly. 

Finally, we note that our algorithm\ \eqref{alg:GibbsMH} above is a sampling method that is essentially identical to the algorithm\ 5 in \cite{neal2000markov}. Neal remarks that it is  a particularly inefficient and naive method of Gibbs sampling for fitting Dirichlet process mixture models: this raises the interesting question of how such an ineffective technique as Gibbs sampling could possibly explain the evolution of complex organisms. 

\section{Conclusion}

Choosing an appropriately simple model is important. To explain the robust success of evolution in developing so many different complex organisms, under so many different conditions of breeding and selection, it seems plausible that there may be some general underlying mathematical principle that gives evolution its learning power.  To find out what that principle may be, it seems reasonable to propose the simplest possible models that may capture what is essential to evolution's computational success, and which leave out complicating details that may be accidents of the biological implementation. Exchangeable breeding models (EBMs)   are one possibility.

As far as we are aware, EBMs are novel in that they are probability models of well known types, for which valid MCMC sampling methods can be interpreted as evolutionary computations.  The stationary distributions of these models can be given in closed form for arbitrary fitness functions. The models are, of course, not fully novel in that similar models were originally developed in mathematical genetics; our contribution is simply to point out that these models can be extended and used for evolutionary computation instead of genetic analysis. 

A number of  questions suggest themselves. Does EB capture the evident learning ability of sexual reproduction, or does bi-parental reproduction, with recombination of linear chromosomes, offer some essential advantage? 
What are appropriate exchangeable breeding distributions for complex machines: in particular, for collections of interacting genes? 
Are there `super-evolutionary' MCMC or variational algorithms that would converge to the stationary distribution faster than simple Gibbs-within-Metropolis sampling?

In the EBM models, sexual evolution is modelled as Gibbs sampling with an additional selection mechanism that preserves detailed balance. There is now extensive experience in machine learning of using MCMC methods, and this experience has been that Gibbs sampling is typically slow for high-dimensional distributions.   An obvious question is whether Gibbs sampling becomes `fast' for high dimensional evolutionary models.  Alternatively, are there completely different phenomena -- such as adaptive ecological competition between species -- that are needed to explain the effectiveness of evolution  of complex organisms?

\section*{Acknowledgements}
We would like to thank Prof Byoung-Tak Zhang, Sang-Woo Lee, and Prof Bert Kappen for helpful and stimulating conversations. This work was partly carried out while Chris Watkins was visiting the Biointelligence Group at Seoul National University, with the support of the MSIP (Ministry of Science, ICT \& future Planning), Republic of Korea.

\bibliographystyle{plain}
\bibliography{superevolution}

\section*{Appendix A}

\newcommand{\twobytwo}[4]{
\begin{array}{cc}
#1  & #2    \\
#3  & #4     
\end{array}
}

\newcommand{\onebytwo}[2]{
\begin{array}{cc}
#1 & #2
\end{array}
}

\newcommand{\tbtA}{\twobytwo{0}{0}{1}{0}}
\newcommand{\tbtB}{\twobytwo{0}{1}{1}{0}}
\newcommand{\tbtC}{\twobytwo{0}{1}{1}{1}}
\newcommand{\tbtD}{\twobytwo{0}{0}{1}{1}}

\newcommand{\gzero}{\onebytwo{0}{0}}
\newcommand{\gone}{\onebytwo{0}{1}}
\newcommand{\gtwo}{\onebytwo{1}{0}}
\newcommand{\gthree}{\onebytwo{1}{1}}

\newcommand{\tbtAtop}{\twobytwo{\cdot}{\cdot}{1}{0}}
\newcommand{\tbtBtop}{\twobytwo{\cdot}{\cdot}{1}{0}}
\newcommand{\tbtCtop}{\twobytwo{\cdot}{\cdot}{1}{1}}
\newcommand{\tbtDtop}{\twobytwo{\cdot}{\cdot}{1}{1}}

\newcommand{\tbtAbottom}{\twobytwo{0}{0}{\cdot}{\cdot}}
\newcommand{\tbtBbottom}{\twobytwo{0}{1}{\cdot}{\cdot}}
\newcommand{\tbtCbottom}{\twobytwo{0}{1}{\cdot}{\cdot}}
\newcommand{\tbtDbottom}{\twobytwo{0}{0}{\cdot}{\cdot}}

\newcommand{\pcond}[2]{P\Bigl( #1 \Bigm \vert #2 \Bigr)}

We prove that the sampling method of section\ \ref{sec:fecundityselection} is not reversible, by presenting a specific counterexample that violates the Kolmogorov cycle condition. Consider a probability model with population size 2, each genome having two loci (elements), and suppose that the possible allele values are 0 or 1. There are four possible genomes $ 00, 01, 10$ and $11$, with fecundities are $w_{00}, w_{01}, w_{10}, w_{11}$ respectively. Let the concentration parameters be $\alpha_0, \alpha_1$, and let $\alpha = \alpha_0 + \alpha_1$.  

The transition probabilities are as follows. A population in which the first genome is $\onebytwo{0}{0}$ and the second genome is $\onebytwo{1}{1}$ is denoted
\begin{equation*}
G = \tbtCtop
\end{equation*}
Suppose that the first genome is Gibbs-sampled. The probability of obtaining $\onebytwo{0}{1}$ is: 
\begin{equation*}
\pcond{\onebytwo{0}{1} }{\tbtC} = \frac{\alpha_0}{w_{11}+\alpha} \frac{w_{11} +\alpha_1}{w_{11}+\alpha}
\end{equation*}
where $\frac{\alpha_0}{w_{11}+\alpha} $ is the probability of sampling a $0$ as the first element, and $\frac{w_{11} +\alpha_1}{w_{11}+\alpha}$ is the probability of sampling a $1$ as the second element.

We consider a specific cycle of four states, shown below, and we compute the product of transition probabilities in clockwise and anti-clockwise directions around the cycle: these probabilities are set out in tables \ref{table:clockwise} and \ref{table:anticlockwise}.

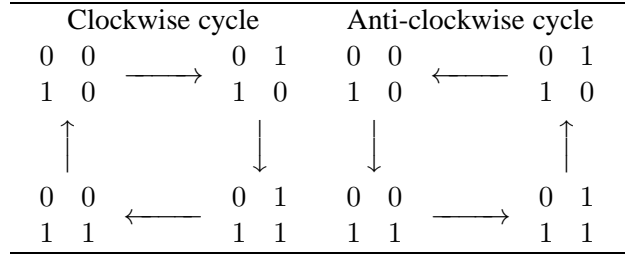
\begin{center}
\begin{tabular}{cc}
\hline
Clockwise cycle & Anti-clockwise cycle\\
$
\begin{CD}
\tbtA @>>>  \tbtB \\
@AAA       @VVV \\
\tbtD @<<< \tbtC
\end{CD}
$
&
$
\begin{CD}
\tbtA @<<<  \tbtB \\
@VVV       @AAA\\
\tbtD @>>> \tbtC
\end{CD}
$\\
\hline
\end{tabular}
\end{center}

\begin{table}[h]
\begin{center}
\begin{tabular}{cl}
Clockwise states& transition probabilities\\
$\tbtA$&\\
$\downarrow$&$\pcond{\gone}{\tbtBtop}= \dfrac{ \alpha_0 \alpha_1} {(w_{10}+\alpha)^2}$\\
$\tbtB$&\\
$\downarrow$&$\pcond{\gthree}{\tbtCbottom}= \dfrac{ \alpha_1 (w_{01} + \alpha_1)} {(w_{01}+\alpha)^2}$\\
$\tbtC$&\\
$\downarrow$&$\pcond{\gzero}{\tbtDtop}= \dfrac{ \alpha_0^2} {(w_{11}+\alpha)^2}$\\
$\tbtD$&\\
$\downarrow$&$\pcond{\gtwo}{\tbtDbottom}= \dfrac{ \alpha_1 (w_{00} + \alpha_0)} {(w_{00}+\alpha)^2}$\\
$\tbtA$&\\
\end{tabular}
\caption{Transition probabilities of clockwise cycle. (We ignore the choice of which genome to resample: these probabilities could be made constant at $\frac{1}{2}$.)}
\label{table:clockwise}
\end{center}
\end{table}%

\begin{table}[h]
\begin{center}
\begin{tabular}{cl}
Anti-clockwise states & transition probabilities\\
$\tbtA$&\\
$\uparrow$&$\pcond{\gzero}{\tbtBtop}= \dfrac{ \alpha_0 (w_{10}+\alpha_1)} {(w_{10}+\alpha)^2}$\\
$\tbtB$&\\
$\uparrow$&$\pcond{\gtwo}{\tbtCbottom}= \dfrac{ \alpha_0 \alpha_1} {(w_{01}+\alpha)^2}$\\
$\tbtC$&\\
$\uparrow$&$\pcond{\gone}{\tbtDtop}= \dfrac{ \alpha_0 (w_{11} + \alpha_1)} {(w_{11}+\alpha)^2}$\\
$\tbtD$&\\
$\uparrow$&$\pcond{\gthree}{\tbtDbottom}= \dfrac{ \alpha_1^2} {(w_{00}+\alpha)^2}$\\
$\tbtA$&\\
\end{tabular}
\caption{Transition probabilities of anti-clockwise cycle. }
\label{table:anticlockwise}
\end{center}
\end{table}%

The ratio of the product of the clockwise transition probabilities to the product of the anti-clockwise probabilities is: 
\begin{equation*}
\frac{ (w_{01} + \alpha_1)(w_{00} + \alpha_0) }{(w_{10} + \alpha_1) ( w_{11} + \alpha_1)} 
\end{equation*}
Since $w_{00}, w_{01}, w_{10},$ and $w_{11}$ are arbitrary fitness values, this ratio is not in general equal to 1, and the Kolmogorov cycle condition does not hold.

\end{document}